\documentclass[10pt, conference, letterpaper]{IEEEtran}

\usepackage{cite}
\ifCLASSINFOpdf
  \usepackage[pdftex]{graphicx}
\else
  \usepackage[dvips]{graphicx}
\fi

\usepackage{enumitem}
\setlist[enumerate]{leftmargin= 0.5 cm}
\setlist[itemize]{leftmargin=0.3 cm}

\usepackage{epstopdf}
\usepackage{graphicx}
\usepackage{subfigure}
\usepackage{amsmath}
\usepackage{amsthm}
\usepackage{algorithm}
\usepackage{algorithmic}
\usepackage{amssymb}
\usepackage[mathscr]{eucal}
\usepackage{url}
\usepackage{thmtools}
\usepackage{color}

\usepackage{multirow}

\usepackage{lipsum}
\usepackage{capt-of}
\usepackage{soul} 

\theoremstyle{definition}
\newtheorem{defi}{Definition}
\newtheorem{lemma}{Lemma}
\newtheorem{theo}{Theorem}
\newtheorem{coro}{Corollary}

\declaretheorem[style=definition,name=Example,qed=$\blacksquare$]{Example}

\DeclareMathOperator*{\argmax}{arg\,max}

\DeclareMathAlphabet{\mathpzc}{OT1}{pzc}{m}{it} 

\makeatletter
\newcounter{procedure}
\newfloat{procedure}{htbp}{loa}[procedure]
\floatname{procedure}{Procedure}

\makeatother



\IEEEoverridecommandlockouts
\def\BibTeX{{\rm B\kern-.05em{\sc i\kern-.025em b}\kern-.08em
    T\kern-.1667em\lower.7ex\hbox{E}\kern-.125emX}}
\begin{document}
\title{Online Multicast Traffic Engineering for Software-Defined Networks 
}
\author{Sheng-Hao Chiang\IEEEauthorrefmark{1}, Jian-Jhih Kuo\IEEEauthorrefmark{1}, Shan-Hsiang Shen\IEEEauthorrefmark{2}, De-Nian Yang\IEEEauthorrefmark{1}, and Wen-Tsuen Chen\IEEEauthorrefmark{1}\IEEEauthorrefmark{4}\\
\IEEEauthorrefmark{1}Inst. of Information Science, Academia Sinica, Taipei, Taiwan\\
\IEEEauthorrefmark{2}Dept. of Computer Science \& Information Engineering, \\National Taiwan University of Science \& Technology, Taipei, Taiwan\\
\IEEEauthorrefmark{4}Dept. of Computer Science, National Tsing Hua University, Hsinchu, Taiwan\\
\{jiang555,lajacky\}@iis.sinica.edu.tw, sshen@csie.ntust.edu.tw, \{dnyang,chenwt\}@iis.sinica.edu.tw
\thanks{
This work is supported by Ministry of Science and Technology of Taiwan under contracts No. MOST 105-2622-8-009-008, No. MOST 105-2221-E-001-001- and No. MOST 105-2221-E-001-015-MY3 and by Academia Sinica Thematic Research Grant.
}
}
\maketitle

\begin{abstract}
Previous research on SDN traffic engineering mostly focuses on static traffic, whereas dynamic traffic, though more practical, has drawn much less attention. Especially, online SDN multicast that supports IETF dynamic group membership (i.e., any user can join or leave at any time) has not been explored. Different from traditional shortest-path trees (SPT) and graph theoretical Steiner trees (ST), which concentrate on routing one tree at any instant, online SDN multicast traffic engineering is more challenging because it needs to support dynamic group membership and optimize a sequence of correlated trees without the knowledge of future join and leave, whereas the scalability of SDN due to limited TCAM is also crucial. In this paper, therefore, we formulate a new optimization problem, named Online Branch-aware Steiner Tree (OBST), to jointly consider the bandwidth consumption, SDN multicast scalability, and rerouting overhead. We prove that OBST is NP-hard and does not have a $|D_{max}|^{1-\epsilon}$-competitive algorithm for any $\epsilon >0$, where $|D_{max}|$ is the largest group size at any time. We design a $|D_{max}|$-competitive algorithm equipped with the notion of the budget, the deposit, and Reference Tree to achieve the tightest bound. The simulations and implementation on real SDNs with YouTube traffic manifest that the total cost can be reduced by at least 25\% compared with SPT and ST, and the computation time is small for massive SDN.
\end{abstract}

\section{Introduction}
Software-defined networking (SDN) provides a promising architecture for flexible network resource allocations to support a massive amount of data transmission \cite{OpenFlow}. SDN separates the control plane and data plane, allowing the control plane to be programmable to efficiently optimize the network resources. OpenFlow \cite{OpenFlow} in SDN includes two major components: controllers (SDN-Cs) and forwarding elements (SDN-FEs). Controllers derive and install forwarding rules according to different policies in the control
plane, whereas {forwarding elements in switches} then deliver packets according to the forwarding rules. Compared with the current Internet techniques/multicast, routing paths no longer need to be the shortest ones, and the paths can be embedded more flexibly inside the network. Since SDN provides a better overview of network topologies and enables centralized computation, traditional theoretical studies on SDN traffic engineering mostly focus on static traffic flows. Recently, allocations of online dynamic unicast have drawn increasing attention \cite{DynamicRouting-SDN,2SegmentRouting}.
However, online multicast traffic engineering that supports \emph{IETF dynamic group membership} \cite{RFC2236} in the current Internet multicast has not been explored for SDN\footnote{{\color{black}To implement multicast, the SDN controller installs a rule with multiple output ports in its action field into each branch node. Then, the branch node clones and forwards corresponding packets to the output ports}}. Dynamic group membership plays a crucial role in practical multicast services because it allows any user to join and leave a group at any time (e.g., for a conference call or video broadcast).

Compared to unicast, multicast has been demonstrated in empirical studies to effectively reduce the overall bandwidth consumption by around 50\% \cite{BenefitsOfMulticast}. It employs a multicast tree, instead of disjoint unicast paths, from the source to span all destinations of a multicast group. The current Internet multicast standard, i.e., PIM-SM \cite{RFC4601}, employs a shortest-path tree (SPT), and therefore does not support traffic engineering since the path from the source to each destination is fixed to be the shortest one. SPT tends to lose many good opportunities to reduce the bandwidth consumption by sharing more common edges among the paths to different destinations. In contrast, a Steiner tree (ST) \cite{STProblems} in Graph Theory provides flexible routing to minimize the bandwidth consumption (i.e., total number of edges). 
Nevertheless, ST is not designed for online multicast with dynamic group membership because 1) re-computing the Steiner tree after any joins or leaves is computationally intensive, especially for a large group, and 2) the overhead (e.g., installing new rules to modify Group Table in OpenFlow) incurred from rerouting the previous Steiner tree to the latest one is not minimized.

Moreover, the scalability problem for SDN multicast is much more serious since for each network with $n$ nodes, the number of possible multicast groups is $O(2^n)$, whereas the number of possible unicast connections is only $O(n^2)$. Therefore, it is more difficult for SDN-FEs to store the forwarding entries of all multicast groups in OpenFlow Group Table due to the limited TCAM size\footnote{{\color{black}SDN switches match more fields in packets that implies that the length of each rule entry is much longer than traditional switches, so SDN switches can support fewer rule entries with the same TCAM size.}}  \cite{DIFANE,multicastTEforSDN}. To remedy this issue, a promising way is to exploit the branch forwarding technique \cite{UnicastTunnelForMulticast,UnicastTunnelForMulticast2,UnicastTunnelForMulticast3}, which stores the multicast forwarding entries in branch nodes (with at least three incident edges), instead of every node of a tree. It can effectively remedy the multicast scalability problem since packets are forwarded in a unicast tunnel from the logic port of a branch node in SDN-FE \cite{OpenFlow} to another branch node. All nodes in the unicast tunnel are no longer necessary to maintain the forwarding entry of the multicast group. Nevertheless, the number and positions of branch nodes need to be carefully selected, but those important factors are not examined in SPT and ST.

Different from traditional traffic engineering for static traffic, for \emph{scalable online multicast} in SDN, it is necessary to minimize not only the tree size and branch node number of a tree at any instant but also the rerouting overhead\footnote{{\color{black}In SDN architecture, a controller reroutes traffic by sending control messages to update forwarding rules in switches. This procedure creates the following overheads:1) the controller consumes CPU power to process the control messages, 2) the bandwidth between the controller and switches, and 3) switches need to decode the control messages and the CPUs in switches are usually wimpy~\cite{devoflow}.}} to tailor a sequence of trees over time, while the future user joining and leaving are unknown during the rerouting. In this paper, therefore, we formulate a new optimization problem, named Online Branch-aware Steiner Tree (OBST), for SDNs with dynamic group membership to consider the above factors. We prove that OBST is NP-hard. Note that traditional ST problem is also NP-hard but approximable within 1.39 \cite{SteinerTreeBestRatio}. However, OBST is more challenging because all above costs from a sequence of correlated trees need to be carefully examined. Indeed, we prove that for any $\epsilon >0$, there is no $|D_{max}|^{1-\epsilon}$-competitive algorithm for OBST unless P $=$ NP, where $D_{max}$ is the largest destination set at any time. For OBST, we design a $|D_{max}|$-competitive algorithm,\footnote{Theoretically, a 
$c$-competitive algorithm \cite{OnlineCompetitiveAnalsis} is an online algorithm that does not know the future information at any time, and the performance gap between the algorithm and the optimal offline algorithm (i.e., an oracle that always know the future) is at most
$c$. In this paper, the future information is the dynamic group membership in the future.} called Online Branch-aware Steiner Tree Algorithm (OBSTA)\footnote{{\color{black}A rule is installed in each edge switch to redirect IGMP join messages from destinations to an SDN controller. Then, OBSTA on the top of the SDN controller connects the destinations to a multicast tree by asking the SDN controller to install forwarding rules into switches in the tree.}}, to achieve the tightest bound. To carefully examine the temporal correlation (which has not been exploited in all static multicast algorithms) of trees, we introduce the notion of \emph{budget} and \emph{deposit}, and a larger budget and deposit allows OBSTA to reroute the tree more aggressively for constructing a more bandwidth efficient tree. Moreover, we construct Reference Tree for OBSTA to guide the rerouting process for achieving the performance bound. {\color{black}Simulation and implementation in a real SDN with YouTube traffic manifest that}
OBSTA can reduce the total cost by 25\% compared with SPT and ST.



The rest of this paper is organized as follows. Section \ref{sec: related work} introduces the related work. Sections \ref{sec: problem formulation} and \ref{sec: HardnessResults} formally define OBST and describe the hardness results. We present the algorithm design of OBSTA in Section \ref{sec: algorithm}, and Section \ref{sec: evaluation} shows the simulation and implementation results on real topologies. Finally, Section \ref{sec: conclusion} concludes this paper.

\section{Related Work}\label{sec: related work}
Traffic engineering for unicast in SDN has attracted a wide spectrum of attention. Mckeown et al. \cite{OpenFlow} study the performance of OpenFlow in heterogeneous SDN switches. Jain et al. \cite{B4} develop private WAN of Google Inc. with the SDN architecture.
Agarwal et al. \cite{TE-SDN} present unicast traffic engineering in an SDN network with only a few SDN-FEs, while other routers followed a standard routing protocol, such as OSPF. By leveraging SDN and reconfigurable optical devices, Jin et al. \cite{BulkTransferSDOpticalWAN} design centralized systems to jointly control the network layer and the optical layer, while Jia et al. \cite{MinMakespanUnicastSDOpticalWAN} present approximation algorithms for transfer scheduling. Gay et al. \cite{QuickReactUnexpectedChanges} study traffic engineering with segment routing in wide area networks with network failures. Qazi et al. \cite{SIMPLE-fying} design a new system to effectively control the SDN middleboxes. Cohen et al. \cite{TableUtilization-SDN} maximize the network throughput with limited sizes on forwarding tables. However, the above studies only focus on  unicast traffic engineering, and multicast traffic engineering in SDN has attracted much less attention.

For multicast, the current standard IETF PIM-SM \cite{RFC4601}, which is designed to support the dynamic multicast group membership in IETF IGMP \cite{RFC2236}, relies on unicast routing protocols to discover the shortest paths from the source to the destinations for building a shortest-path tree (SPT). However, SPT is not designed to support traffic engineering. Although the Steiner tree (ST) \cite{STProblems} minimizes the tree cost, ST is computationally intensive and is not adopted in the current Internet standard. Overlay ST \cite{OverlayST1,OverlayST2}, on the other hand, is an alternative to constructing a bandwidth-efficient multicast tree in the P2P environment. However, the path between any two P2P clients is still a shortest path in Internet. Recently,
Huang et al. optimized the routing of multicast trees in SDN \cite{ScalableMulticastforSDN,reliableMulticastforSDN,multicastTEforSDN}.
Zhu et al. \cite{Multicast-VirtualDCNetwork} study the multicast tree with virtualized software switches. However, the above studies are designed for static multicast trees at any instant, instead of optimizing a sequence of dynamic trees with varying destinations. 

For dynamic traffic, dynamic unicast has been explored in \cite{DynamicRouting-SDN,CompetitiveRouting,failure_recovery_rerouting,power_saving_rerouting} to rerouting the traffic for failure recovery and power saving (i.e., to turn off more links and nodes). Previous literature on dynamic multicast focuses on protocol and architecture design, instead of theoretical study, for other kinds of networks. Baddi et al.\cite{dynamic_multicast_mobile_ip} revise PIM-SM to support multicast join, leave, and handover in mobile wireless networks. Markowski~\cite{dynamic_multicast_optical_networks} considers spectrum and modulation setup for dynamic multicast in optical networks. Xing et al. ~\cite{dynamic_multicast_network_coding} explore dynamic multicast with network coding to minimize the coding cost. However, the dynamic group membership in IETF IGMP has not been considered in any SDN theoretical research for multicast traffic engineering.

\section{Problem Formulation} \label{sec: problem formulation}
In this paper, we formulate a new online multicast problem, namely \emph{Online Branch-aware Steiner Tree (OBST)}. It jointly considers tree routing, rerouting, and scalability in SDNs. More specifically, the network is modeled as a weighted graph $G=\{V,E\}$ with a source $s\in V$ and a set of candidate destinations $D\subseteq V$, where {\color{black}$w(e)$ is the weight of edge $e\in E$, respectively.}
{\color{black}
A larger weight can be assigned to 
a link with the higher delay and higher loss rate (due  to congestion) \cite{LinkWeightSetting
}.}
For online multicast, each destination is allowed to join and leave at the beginning of each time slot $t_i$,\footnote{
We assume that time is divided into discrete time slots, and each user request arrives in a certain time slot (i.e., a short duration) \cite{timeSlotAssumption,timeSlotAssumption2}. This assumption is reasonable in networks where connections are active for a time period of seconds, minutes, or even hours. The size of the time slot can be scaled according to the different settings and analyses. Delay sensitive applications such as video conference require shorter time duration. Other applications such as web browsing can tolerate longer time duration.} and let $D_i$ denote the set of destinations during $t_i$, where $D_i\subseteq D$ and $D_{max}=\argmax_{D_i}\{|D_i|\}$. At $t_i$, the SDN controller builds a multicast tree $T_i$ by tailoring $T_{i-1}$ to connect the source $s$ and the destinations in $D_i$, where the future information $D_j$ is unknown for every $j>i$. The tree cost $w(T_i)$ is the sum of edge weights in the multicast tree $T_i$ (i.e., $w(T_i)=\sum_{e\in T_i} w(e)$). To consider the scalability of SDN (i.e., TCAMs), {\color{black}
the branch cost $b_{T_i}$ in $t_i$ is 
the total number (or the total weighted number) of branch nodes in $T_i$,}
where each branch node has at least three neighbors in $T_i$.\footnote{Here the source is regarded as a branch node since it has to maintain Group Table in OpenFlow.} Moreover, at $t_i$, for the destinations in $t_{i-1}$ that will also stay in this time slot, it is necessary to reroute the tree spanning them because the tree routing $T_{i-1}$ may no longer be bandwidth efficient. To consider the rerouting overhead, we first define the operation \emph{pruning} of a tree and the symmetric difference of two graphs, and then define the rerouting cost $rc_i$ in $t_i$ according to \cite{ConsistentMigration-SDN} .

\begin{defi}
Given a node set $L\subseteq D$ and a tree $T$ spanning the destinations in $D$, the operation, \emph{pruning}, denoted by $T\ominus L$, outputs a maximum subtree of $T$, where each leaf is either a destination in $D\setminus L$ or the source of $T$.
\end{defi}

\begin{defi}
Given two graphs $G_1$, $G_2$, the symmetric difference $G_1\triangle G_2$ of the two graphs is the graph with the edges in $G_1\cup G_2$ but not in $G_1\cap G_2$.
\end{defi}

In other words, $T\ominus L$ considers the case that destination set $L$ leaves the tree, and the path connecting each $d$ in $L$ to the closest branch node is thereby removed from $T$. Moreover, each edge in $G_1\triangle G_2$ appears in either $G_1$ or $G_2$.


\begin{defi}
The rerouting cost in $t_i$ is the total edge cost of the symmetric difference $(T_{i-1}\ominus (D_{i-1}\setminus D_{i}) ) \triangle (T_i\ominus (D_{i} \setminus D_{i-1}))$.
In other words, $rc_i=\sum_{e\in (T_{i-1}\ominus (D_{i-1}\setminus D_{i})) \triangle (T_i\ominus(D_{i}\setminus D_{i-1}))}w(e)$.
\end{defi}

The first part of the symmetric difference represents the old tree without the leaving destinations, and the second part is the new tree without the joining users. Therefore, the symmetric difference is the tree topology change (i.e., rerouting efforts) for the destinations staying in both time slots.


\begin{Example}
Fig. \ref{ProblemExample} presents an example of the rerouting cost for $d_1,d_2\in D_{i-1}$ and $d_2,d_3\in D_i$. Figs. \ref{ProblemExample}(a) and \ref{ProblemExample}(d) show $T_{i-1}$ and $T_i$, respectively. Figs. \ref{ProblemExample}(b) and \ref{ProblemExample}(c) show the old tree after $d_1$ leaves and the new tree if $d_3$ was not a destination, respectively.
Fig. \ref{ProblemExample}(e) shows the symmetric difference with $rc_i=23$. It reroutes $d_2$ by removing the old path in Fig. \ref{ProblemExample}(b) and adding the new path in Fig. \ref{ProblemExample}(c).
\end{Example}

\begin{figure}[tb]
\centering
\includegraphics[scale=0.35] {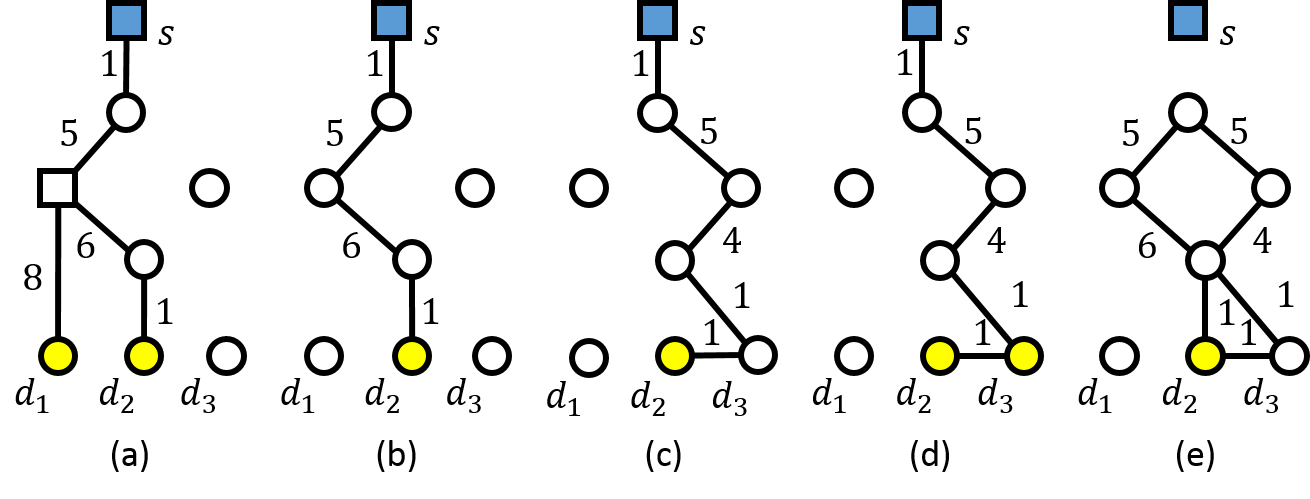}
\vspace{-1.5 mm}
\caption{(a) $T_{i-1}$ for $d_1$ and $d_2$.  (b) $T_{i-1}\ominus \{d_1\}$. (c) $T_i\ominus \{d_3\}$. (d) $T_{i}$ for $d_2$ and $d_3$. (e) Symmetric difference of $T_{i-1}\ominus \{d_1\}$ and $T_i\ominus \{d_3\}$.}
\label{ProblemExample}
\end{figure}

\begin{defi}
Given $G=\{V,E\}$, a source $s\in V$, and a collection of destination sets, $\{D_1,D_2,...,D_{n}\}$ for time $t_1,t_2,...,t_n$, the OBST problem finds a set of trees $T=\{T_{1},...,T_{n}\}$ such that 1) $s$ is connected to all destinations in $D_i$ via $T_i$ for every $t_i$, and 2) the total cost
$\sum_{i=1}^{n} (w(T_{i})+\alpha \times b_{T_{i}}+\beta \times rc_{i})$ is minimized, where $\alpha\geq 0$ and $\beta\geq 0$. In the \emph{offline} OBST problem, $D_j$ is available at any $t_i$ with $i<j$, but the above information is not available in the \emph{online} OBST problem studied in this paper.
\end{defi}

Therefore, OBST aims to find a series of multicast trees such that the weighted sum of the overall tree costs, the branch costs, and rerouting costs is minimized, where $\alpha$ and $\beta$ are tuning knobs \cite{ScalableMulticastforSDN,reliableMulticastforSDN,multicastTEforSDN} for network operators to differentiate the importance of network load, scalability, and rerouting overhead. Compared with the SPT and ST, OBST not only provides the flexible routing of trees but also carefully examines the practical issues of OpenFlow in dynamic SDNs. 

Fig. \ref{Comparison} presents an example to compare SPT, ST, and OBST with one source $s$ and two destinations $d_1$ and $d_2$ in Fig. \ref{Comparison}(a), where $D_1=\{d_1\}$, $D_2=\{d_1,d_2\}$, $\alpha=1$, and $\beta=0.2$. SPT, ST, and OBST are the same in $t_1$ with the total cost as $(7.5+2.5)+1\times 1+0.2\times 0=11$ in Fig. \ref{Comparison}(b). However, after $d_2$ joins, SPT becomes $(7.5+2.5+6.2+4)+1\times 1+0.2\times 0=21.2$ in Fig. \ref{Comparison}(c), and ST becomes $(6.2+4+4)+1\times 2+0.2\times (7.5+2.5+6.2+4)=20.24$ in Fig. \ref{Comparison}(d). In contrast, OBST only incurs $(7.5+2.5+2+4)+1\times 2+0.2\times 0=18$ in Fig. \ref{Comparison}(e). Note that OBST does not choose an alternative $(7.5+2.5+6)+1\times 1+0.2\times 0=17$ in Fig. \ref{Comparison}(f) because if $d_1$ leaves the tree, $d_2$ is inclined to induce dramatic rerouting.

\begin{figure}[tb]
\centering
\includegraphics[scale=0.35] {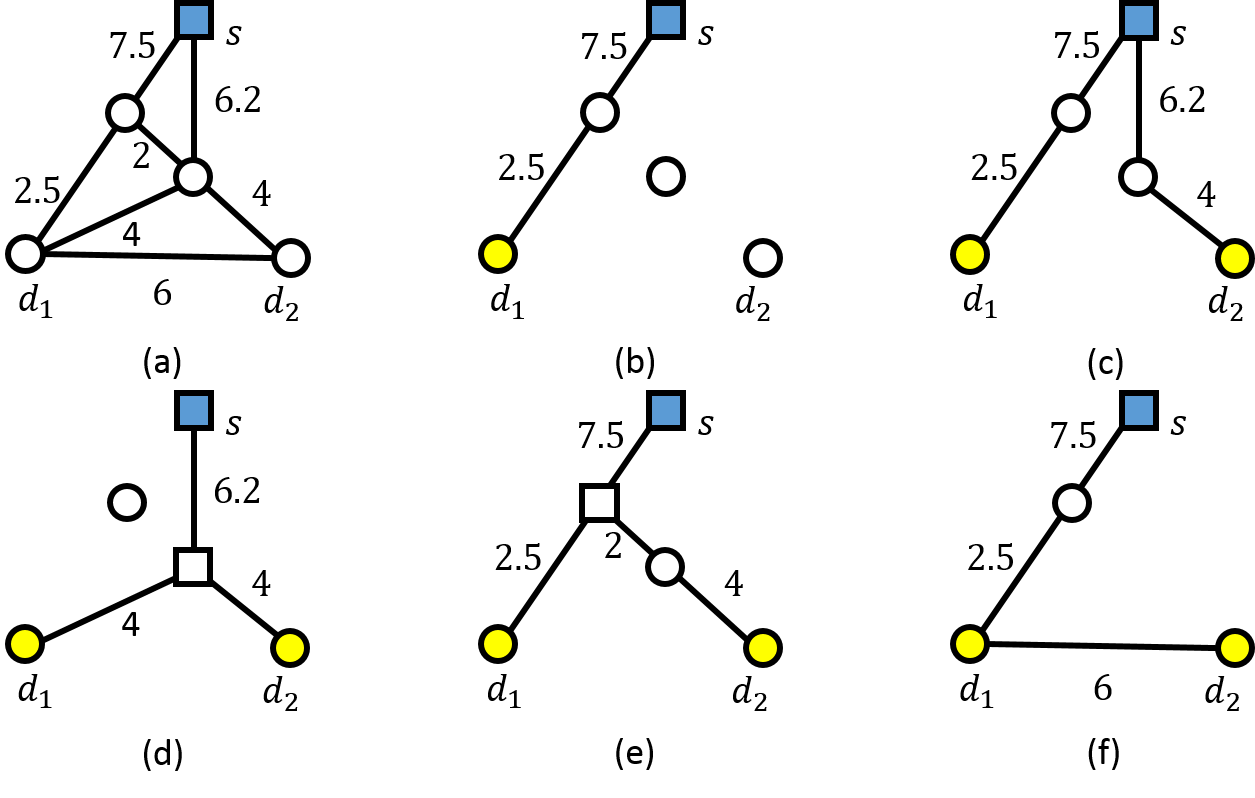}
\vspace{-1.5 mm}
\caption{(a) Input network $G$. (b) $T_1$ of SPT, ST, and OBST. (c) $T_2$ of SPT. (d) $T_2$ of ST. (e) $T_2$ of OBST. (f) Alternative $T_2$ of OBST.}
\label{Comparison}
\end{figure}


\section{Hardness Results\label{sec: HardnessResults}}

The OBST problem is more challenging than the Steiner tree problem since ST does not consider the SDN scalability and examine the temporal correlation of trees. In the following, we first show that OBST is very challenging in Complexity Theory by proving that the offline OBST problem does not have a $|D_{max}|^{1-\epsilon}$-approximation algorithm for any $\epsilon>0$,\footnote{By contrast, ST can be approximated within $1.39$ \cite{SteinerTreeBestRatio}.} with a gap-introducing reduction from the Hamiltonian Path (HP) problem, where $D_{max}$ is the largest destination set at any time. Afterward, we show that the online OBST problem does not have a $|D_{max}|^{1-\epsilon}$-competitive algorithm.


\begin{defi} \label{defi: Hamiltonian path}
Given any graph $G_H=\{V_H,E_H\}$ and a vertex $y\in V_H$, the HP problem finds a path starting from $y$ to visit every other vertex exactly once.
\end{defi}

\begin{theo} \label{theo: hardness}
For any $\epsilon >0$, no $|D_{max}|^{1-\epsilon}$-approximation algorithm exists for the offline OBST problem unless P $=$ NP.
\end{theo}

\begin{proof}
We prove the theorem with a gap-introducing reduction from the HP problem. For any instance $G_H$, we construct an OBST instance $G$ with the following goals: 1) if $G_{H}$ has a Hamiltonian path starting at $y$, $\mathrm{OPT}(G)\leq2\alpha m( m^{p}+1)$; 2) otherwise, $\mathrm{OPT}(G)>2\alpha m( m^{p}+1) |D_{max}|^{1-\epsilon}$, where  $m$ is the number of nodes in $G_{H}$ and $\mathrm{OPT}(G)$ is the optimal OBST. Specifically, we first create a source node $s$ and $m^p$ clones of $G_{H}$, namely $G_{H,1},G_{H,2},...,G_{H,m^p}$, and add them into $G$, where $p$ denotes the smallest integer greater than $\frac{1-\epsilon+\log_{m}2}{\epsilon}$. Each clone node of $y$ is then connected to $s$ in $G$. The cost of every edge in $G$ is $1$, and $\alpha$ and $\beta$ are set to be any number no smaller than $\frac{m^p(m^{p+1}+1)}{2}$ and $0$, respectively. 

Every node (except $s$) in $G$ joins and leaves the multicast group as follows. 
We first generate an arbitrary permutation $\mathcal{R}$ for all nodes in $G_{H}$. Let $\mathcal{R}[j]$ denote the $j$-th node in $\mathcal{R}$, where $1\leq j \leq m$.
Then, in each $t_i$ with $1\leq i \leq m^{p+1}$, the clone node in $G_{H,\lceil i/m \rceil}$ of $\mathcal{R}[((i-1)\mod m)+1]$ will join the multicast group.
In $t_{m^{p+1}}$, all nodes are in the multicast group, and each node then leaves the group one-by-one in 
each subsequent time slot after $t_{m^{p+1}+1}$
in a reverse order (i.e., first-in-last-out). The total number of time slots is $2m^{p+1}$.


We first prove the first goal (i.e., the sufficient condition). If $G_{H}$ has a Hamiltonian path starting at $y$, there is a path starting from $s$ that visits every destination in $D_m=G_{H,1}$ exactly once in $t_m$, and the source must be in the tree such that $b_{T_i}\geq 1$. That is, the cost of the above path is at most $m+\alpha$. Since $\beta=0$, the smallest cost in $t_m$ is at most $m+\alpha$. Similarly, the smallest costs in $t_{2m}$-th, $t_{3m}$-th, ..., $t_{m^{p+1}}$-th are at most $2m + \alpha$, $3m+\alpha$, ..., $m^{p+1}+\alpha$, respectively. Moreover, because the group size grows until $t_{m^{p+1}}$, the smallest cost in $t_i$ must be no greater than that of $t_{(\lceil i/m\rceil \times m)}$, where $1\leq i\leq m^{p+1}$. Thus, let $\mathrm{OPT}(G_i)$ denote the smallest cost of $t_i$,
\begin{align*}
\mathrm{OPT}(G) &= \sum_{i=1}^{2m^{p+1}}\mathrm{OPT}(G_{i}) = 2\sum_{i=1}^{m^{p+1}}\mathrm{OPT}(G_{i}) \\
                &\leq 2\big[(1+2+...+m^p)m^2+m^{p+1}\alpha\big]\\
                &= m^{p+1}(m^{p+1}+1) +2\alpha m^{p+1}\leq 2\alpha m(m^{p}+1).
\end{align*}

We then prove the second goal (i.e., the necessary condition). If $G_{H}$ does not have a Hamiltonian path starting at $y$, for each time $t_i$ with $1\leq i < m$, the cost is at least $i+\alpha$. In $t_m$, there exists at least one branch node, and $\mathrm{OPT}(G_{m}) \geq m+2\alpha$. Similarly, the smallest cost in $t_i$ must be at least $i+(\lfloor i/m \rfloor+1) \times \alpha$, where $1\leq i\leq m^{p+1}$. Therefore,
\begin{align*}
&\mathrm{OPT}(G)  = 2\sum_{i=1}^{m^{p+1}}(\mathrm{OPT}(G_{i}))\\
                &\geq 2\sum_{i=1}^{m^{p+1}}(i+(\lfloor i/m \rfloor+1) \times \alpha)\\
            & \geq  2m^{p}(\frac{m(m+1)}{2}+\alpha m)+m^{p+1}(m^{p}-1)(m+\alpha)\\
            & > \alpha m^{p+1}(m^{p}+1)=2\alpha m(m^{p}+1)m^{p-\log _{m}2} \\
            & = 2\alpha m( m^{p}+1) m^{(p-p\epsilon)+(p\epsilon -\log_{m}2)} \\
            & \geq  2\alpha m(m^{p}+1) m^{(p+1)(1-\epsilon)}=2\alpha m(m^{p}+1)|D_{max}|^{1-\epsilon}.
\end{align*}
Since $\epsilon $ can be arbitrarily small, for any $\epsilon >0$, there is no $%
|D_{max}|^{1-\epsilon }$-approximation algorithm unless P$=$NP.
\end{proof}

Since any $c$-competitive algorithm for an online problem is also a $c$-approximation algorithm for the offline problem \cite{OnlineCompetitiveAnalsis}, we have the following corollary.
\begin{coro}
For any $\epsilon>0$, no $|D_{max}|^{1-\epsilon}$-competitive algorithm exists for the online OBST problem unless P = NP.
\end{coro}



\section{Algorithm Design of OBSTA} \label{sec: algorithm}
In the following, we propose a $|D_{max}|$-competitive algorithm, called Online Branch-aware Steiner Tree Algorithm (OBSTA), for OBST. Different from SPT and ST that focus on a static tree, we introduce the notion of \emph{budget} $B_{i}$ and \emph{deposit} $dep_{i}$ for each $t_i$ to properly optimize a sequence of correlated trees over time. Intuitively, with a larger budget and deposit, OBSTA is able to reroute the tree more aggressively for constructing a more bandwidth efficient tree in $t_{i+1}$. Nevertheless, the rerouting cost also grows in this case, and it is thus necessary to carefully set a proper budget in order to achieve the performance bound (i.e., competitive ratio). Also, the remaining budget $B_{i}$ can be saved to deposit $dep_{i}$, so that $B_{i+1}+dep_{i}$ can be exploited to reroute the tree in $t_{i+1}$. More specifically, to achieve the competitive ratio, the lower bound of the optimal solution is exploited to derive the budget. The lower bound is the longest shortest path among all source-destination pairs plus the cost of one branch node (proved later). Thus, we set the budget in $t_i$ as $B_{i} = \sum_{d\in D_{i}}dist_{G}(s,d)+\alpha\times |D_{i}|$. On the other hand, the deposit represents the remaining unused budget, i.e., $dep_{i}=dep_{i-1}+B_{i}-w(T_{i})-\alpha\times b_{T_i}-\beta \times rc_{i}$.

Also, we derive a \emph{potential rerouting cost} $prc_{i}$ at each $t_i$ for OBSTA to estimate the worst-case rerouting effort in the future. In the worst case, when all the other destinations in the same subtree of $d$ leave the group, $d$ will connect to $s$ via the shortest path because another alternate zigzag incurs a high cost. Therefore, $prc_{i}=\sum_{e\in \cup_{d\in D_i} (P_{T_{i}}(s,d)\triangle P_{G}(s,d))} w(e)$,
where $P_{A}(s,d)$ denotes the shortest path between $s$ and $d$ in the graph $A$.
OBSTA includes a policy to save a sufficient deposit $dep_{i}\geq \beta \times prc_{i}$ for the future. Moreover, to reduce the rerouting cost, our idea is to organize stable destinations (i.e., the destinations staying with longer time\footnote{Stable users can be extracted by \cite{StablePeer
} according to user statistics, such as the duration and join and leave frequencies.}) in a bandwidth-efficient stable tree, called \emph{Reference Tree (RT)}, and other destinations are then rerouted and attached to RT in OBSTA. Later we show that the budget, deposition, potential rerouting cost, and RT are the cornerstones of OBSTA to achieve the tightest performance bound. Section \ref{sec: evaluation} also manifests that OBSTA outperforms SPT and ST because temporal information is not defined and exploited in SPT and ST.

\subsection{Algorithm Description}

For each time $t_i$, OBSTA includes the following four phases:
1) Reference Tree (RT) Construction, 2) Candidate Deployed Tree (DT) Generation, 3) Candidate DT Patching, and 4) Final DT Selection. More specifically, RT Construction first builds a bandwidth-efficient tree for stable destinations as a reference, whereas DT is the real multicast tree deployed in SDNs. Candidate DT Generate creates several candidate DTs, and based on RT, Candidate DT Patch reroutes and attaches different remaining destinations to create several candidate final solutions, in order to effectively avoid unnecessary zigzag paths in the tree. The final phase then chooses the one with the minimum cost. To achieve the performance bound, it is important for Candidate DT Patching to ensure that the current deposit is sufficient to support rerouting in the future. 
The pseudo-code is presented in Appendix \ref{sec: pseudo-code}.

\subsubsection{Reference Tree (RT) Construction}

At the beginning of each $t_i$, RT Construction identifies a set of stable destinations  \cite{StablePeer
} to construct a reference tree $RT_i$.
Let $\tau_{i}$ denote the number of stable destinations in $RT_i$. In $t_1$ with no destination, only $s$ in included in $RT_1$ with $\tau_{1}=0$. 
Afterward, for each $t_i$ with $i>1$, to achieve the performance bound, OBSTA selects $\tau_{i}$ stable destinations from \cite{StablePeer
} and then builds $RT_{i}$ rooted at $s$ with a subtree of $SPT(D)$ that spans only the $\tau_{i}$ destinations, where $SPT(D)$ is the shortest path tree spanning all candidate destinations in $D$. More specifically, to evaluate the degree of stability for each destination, the \emph{Stability Index} (SI) \cite{StablePeer} of a node $d\in D$ is defined as follows,
\begin{equation}\notag
SI(d)=
\begin{cases}
\begin{array}{ll}
\frac{2h_d}{n-a_d} & \text{if } h_d\leq \frac{n-a_d}{2},\\
1 & \text{otherwise},\\
\end{array}
\end{cases}
\end{equation}
where $n$, $a_d$, and $h_d$ denote the number of time slots for the multicast, the first arrival time slot of $d$, and the duration in the group of $d$, respectively.
A destination $d$ with an $SI(d)>H$ will be selected as the stable destination, where $H$ is a threshold tuned by the network operators.\footnote{By simulations and experiments, \cite{StablePeer} suggests $H$ should be set to $0.2$, which is sufficient to filter out most transient nodes, and thus well predicts the future stable nodes.} Thus, $\tau_i=|\{d| d\in D \wedge SI(d)> H\}|$.


\subsubsection{Candidate Deployed Tree (DT) Generation}

In contrast to RT $RT_i$, DT $T_i$ is the solution tree to be deployed in SDNs at $t_i$. Therefore, for the destinations leaving in the beginning of $t_i$, OBSTA first removes them from $T_{i-1}$ by pruning the edges not on a path between the source and any staying destination, in order to create an initial DT $T_{i,leave}$ ($=T_{i-1}\ominus (D_{i-1}\setminus D_i)$) in $t_i$.
Afterward, OBSTA sorts the destinations in $T_{i,leave}$ in a non-decreasing order of the distance between $s$ and each destination to find the sorted list $(q_{1},q_{2},q_{3},...,q_{k})$, where $k=|D_{i-1}\cap D_{i}|$.
Finally, the candidate DT set is $\{T_{i}^{1}, ..., T_{i}^{k}\}$, where each $T_{i}^{l}$ includes the $l$ closet destinations from the source, (i.e., $(q_{1},...,q_{l})$). For each candidate DT, other destinations not in the tree will be rerouted
and attached to the tree in the next phase as much as possible.
In other words, OBSTA is designed to explore different intermediate trees with varied rerouting strategies to achieve the performance bound.

More specifically, OBSTA derives the routing of each candidate DT by \emph{sprouting}, i.e., $T_{i}^{k}=T_{i,leave}\otimes(q_{1},q_{2},...,q_{k})$ and $T_{i}^{0}={s}$, defined as follows.

\begin{defi}
For any tree $A=\{V_A,E_A\}$ with $V_A$ and $E_A$ as the vertex and edge sets, respectively, let $P_{A}(s,t)$ denote the path from $s$ to $t$ in $A$.
Given an arbitrary sequence $(a_1,a_2,...,a_k)$ of nodes in any subset of $V_A$, where $k\leq |V_A|$, \emph{sprouting} $A\otimes(a_1,a_2,...,a_k)$ outputs a collection of subtrees of $A$, denoted by $(A_{1},..., A_k)$, where each subtree $A_{k}=\bigcup_{l\leq k}P_{A}(s,a_{l})$, and $A_{1}\subseteq A_{2}\subseteq ... \subseteq A_{k}\subseteq A$.
\end{defi}
\begin{figure}[t]
\centering
\includegraphics[scale=0.35] {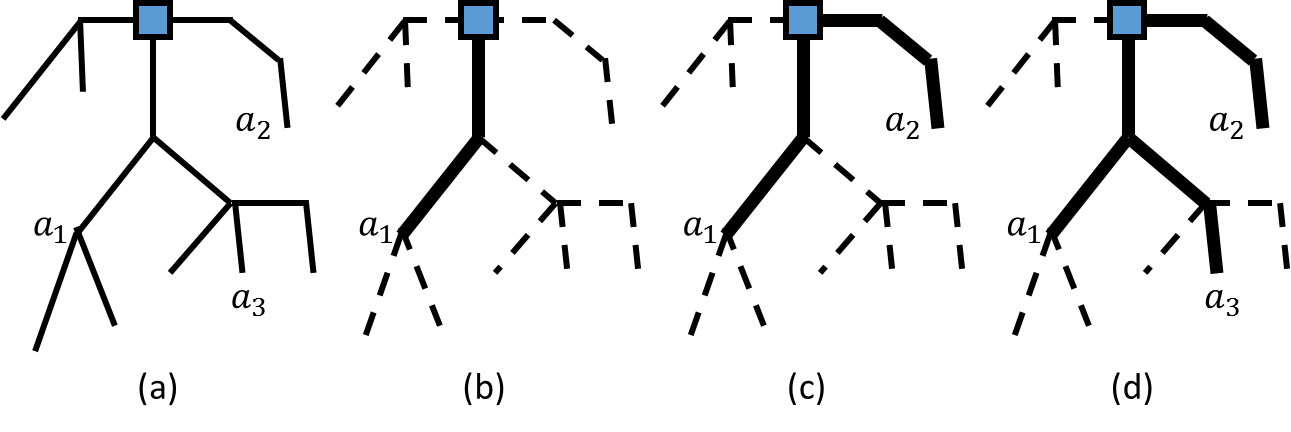}
\vspace{-1.5 mm}
\caption{Example of sprouting $A\otimes (a_{1},a_{2},a_{3})$.}
\label{AotimesB}
\end{figure}
\begin{Example}
Fig. \ref{AotimesB} presents an example of sprouting. Fig. \ref{AotimesB}(a) shows $A$ with a sequence of nodes $(a_{1},a_{2},a_{3})$. The bold trees in Figs. \ref{AotimesB}(b), \ref{AotimesB}(c) and \ref{AotimesB}(d) are subtrees $A_{1}$, $A_{2}$ and $A_{3}$, respectively, where $A_{1}=P_{A}(s,a_{1})$, $A_{2}=P_{A}(s,a_{1})\cup P_{A}(s,a_{2})$ and $A_{3}=P_{A}(s,a_{1})\cup P_{A}(s,a_{2})\cup P_{A}(s,a_{3})$. Fig. \ref{NodeLeave} presents an example of Candidate DT Generation, where Fig. \ref{NodeLeave}(a) shows $T_{i-1}$. After removing the leaving destinations (red nodes), Fig. \ref{NodeLeave}(b) shows $T_{i,leave}$, and the six candidate DTs (i.e., $T_i^0,T_i^1,T_i^2,T_i^3,T_i^4,T_i^5$) are generated, where $T_{i}^{3}$ is shown in Fig. \ref{NodeLeave}(c). 
\end{Example}

\begin{figure}[tb]
\centering
\includegraphics[scale=0.35] {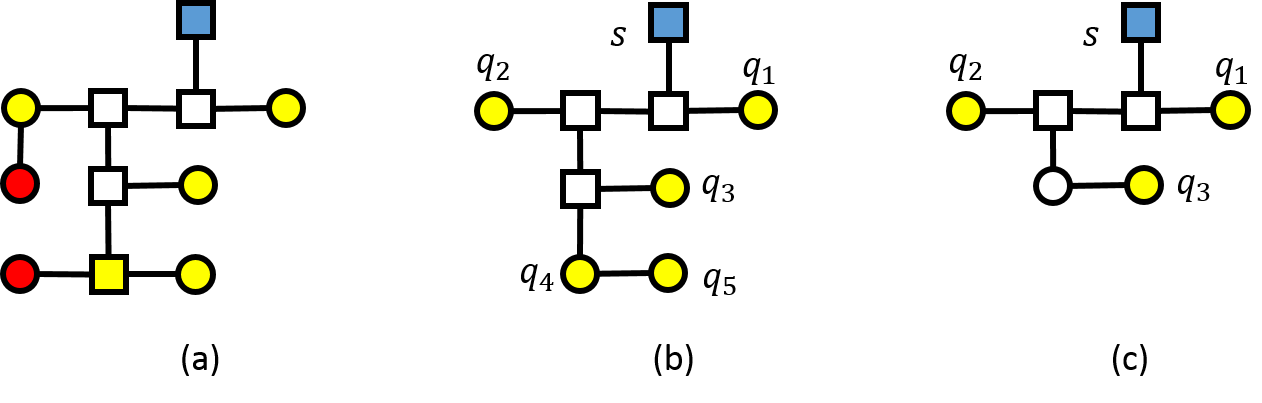}
\vspace{-1.5 mm}
\caption{Example of Candidate DT Generation. (a) $T_{i-1}$. (b) $T_{i,leave}$. (c) $T_{i}^{3}$.}
\label{NodeLeave}
\end{figure}

\subsubsection{Candidate Deployed Tree (DT) Patching}
To generate candidate final solutions, this phase adds new joining destinations and reroutes the destinations (not selected in the previous phase) to the candidate DTs according to RT, whereas the sufficient deposit is required to be followed.
For ease of presentation, we first define the operation, \emph{grafting}, as follows.
\begin{defi} \label{defi: graft}
Given two trees $A$ and $B$ rooted at $s$, \emph{grafting} $A\oplus B$ outputs a tree spanning all destinations in $A$ and $B$ as follows.
1) For each destination $d\in B$ and each shortest path $P_{B}(u,d)$ in $B$ with $u\in A\cap B$, find the smallest one $P_{B}(u_d,d)$ such that every node between $d$ and $u_d$ in $P_{B}(u_d,d)$ is not in $A$.
2) $A\oplus B=A\cup\bigcup_{d\in D_B}P_{B}(u_d,d)$, where $D_B$ denotes the destination set in $B$.

\end{defi}
\begin{figure}[t]
\centering
\includegraphics[scale=0.35] {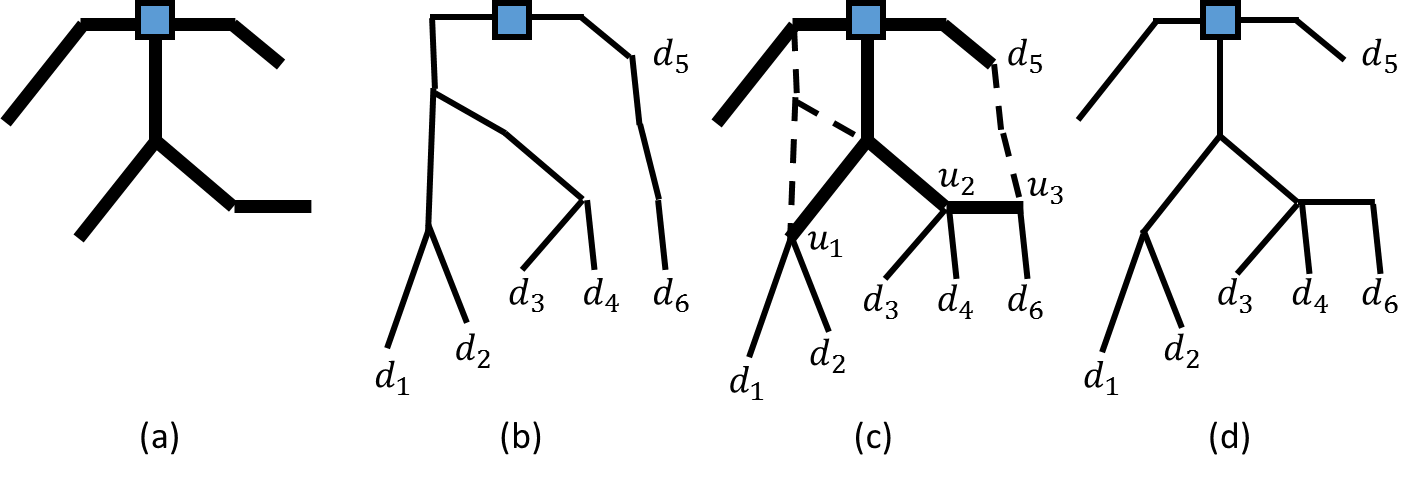}
\vspace{-1.5 mm}
\caption{Example of grafting $A\oplus B$.}
\label{AoplusB}
\end{figure}
\begin{Example}
Fig. \ref{AoplusB} presents an example of grafting. Fig. \ref{AoplusB}(a) shows tree $A$, and Fig. \ref{AoplusB}(b) is $B$ with six destinations. We first find every path $P_{B}(u_d,d)$ for each destination $d$ in $B$ as shown in Fig. \ref{AoplusB}(c). Nodes $d_{1}$ and $d_{2}$ both connect to $u_1$ in $A$, and nodes $d_{3}$ and $d_{4}$ both connect to $u_2$ in $A$. Since node $d_5$ is in $A\cap B$, the path contains only $d_5$. Node $d_{6}$ connects to node $u_3$. Fig. \ref{AoplusB}(d) shows $A\oplus B$, and $A\subseteq A\oplus B$.
\end{Example}
Then, let $D_{i,join}$ denote the set of destinations joining in the beginning of $t_i$. For each candidate DT $T_{i}^{l}$, let $D_{i,reroute}^{l}$ denote the set of staying destinations that are not included in $T_{i}^{l}$ in the previous phase. 
Let $D_{i,attach}^{l}$ denote the destination set containing $D_{i,join}$ and $D_{i,reroute}^{l}$.
Let $D_{i,stable}^{l}$ represent the stable nodes in $D_{i,attach}^{l}$. In the following, we first attach the destinations in $D_{i,stable}^{l}$ to candidate DT $T_{i}^{l}$. Let $RT_{i}(D_{i,stable}^{l})$ denote the minimum-cost subtree of $RT_{i}$ to span $D_{i,stable}^{l}$.
OBSTA first sorts the nodes in $D_{i,stable}^{l}$ in a non-decreasing order of SI (see RT Construction phase)
as $(q_1,q_2,...,q_k)$, where $k=|D_{i,stable}^{l}|$. Then, to patch each candidate DT $T_{i}^{l}$, OBSTA iteratively reroutes and attaches the closest stable destinations according RT 
to generate 
$(T_{i,stable}^{l_{1}},...,T_{i,stable}^{l_{k}})=(RT_{i}(D_{i,stable}^{l})\otimes (q_1,q_2,...,q_k))\oplus T_{i}^{l}$. 
In other words, OBSTA attaches (i.e., grafts) the destinations $(q_1,q_2,...,q_k)$ to each candidate tree $T_{i}^{l}$ according to sprouting of reference tree $RT_i$ and $D_{i,stable}^{l}$. Afterwards, for those $T_{i,stable}^{l_m}$ having sufficient deposit, OBSTA selects the one with the maximum $m$ and set 
the patched candidate DT $T_{i,stable}^{l}$ as $T_{i,stable}^{l_m}$. Note that $T_{i,stable}^{l}=T_{i}^{l}$ if $m=0$. 





\begin{figure}[tb]
\centering
\includegraphics[scale=0.35] {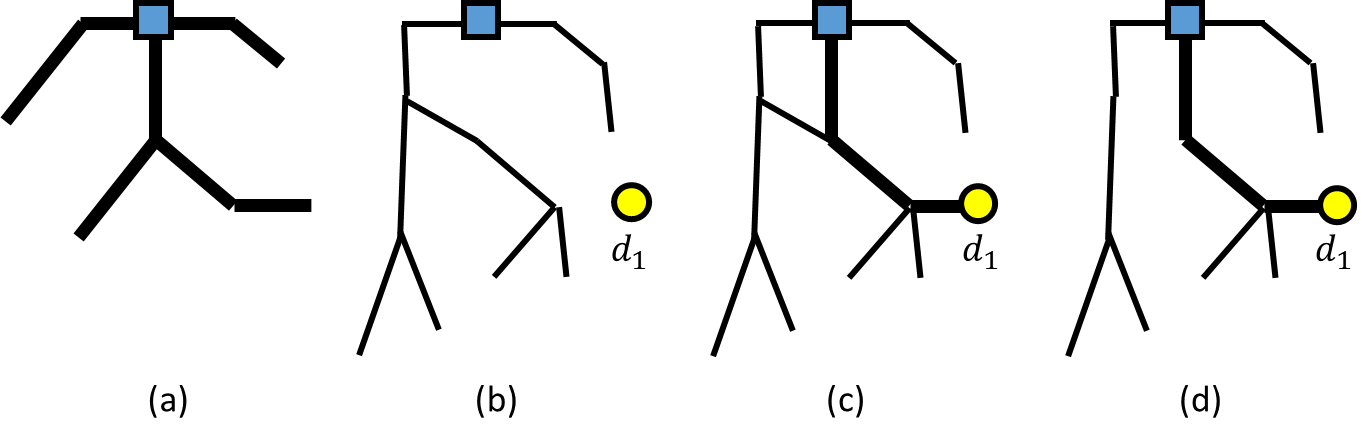}
\vspace{-1.5 mm}
\caption{Example of Candidate DT Patching. (a) the reference tree $RT_{i}$, (b) the current tree and a stable destination node $d_{1}$ that required to be added, (c) attempt to connect $d_{1}$ to the source with the path on reference tree $RT_{i}$, (d) the tree after adding $d_{1}$.}
\label{NodeJoinRef}
\end{figure}

\begin{Example}
Fig. \ref{NodeJoinRef} presents an example of the Candidate DT Patching for stable destination $d_1$, where Figs. \ref{NodeJoinRef}(a) and \ref{NodeJoinRef}(b) are RT and the candidate DT, respectively. OBSTA reroutes $d_1$ to the DT with the path in RT shown in Fig. \ref{NodeJoinRef}(d).
\end{Example}


Afterward, we attach the rest $D_{i,rest}^{l}$ of the destinations (i.e., $D_i-D_{i,stable}^{l}$) to candidate DT $T_{i,stable}^{l}$ with operation \emph{contraction} defined as follows.
\begin{defi}
Given a graph $A=\{V_A,E_A\}$ with a connected subgraph $B=\{V_B,E_B\}$, \emph{contraction} $A\cdot B$ outputs a graph $A'$ by contracting (i.e., merging) the nodes in $B$ into one vertex $u$. For each neighbor $v\in V_A\setminus V_B$ of $u$, $A\cdot B$ includes only the edge with $\min_{e\in E_{u,v}}\{w(e)\}$, where $E_{u,v}$ denotes the set of edges between $u$ and $v$.
\end{defi}
\begin{Example}
Fig. \ref{NodeFinalJoin} presents an example of contraction.
Fig. \ref{NodeFinalJoin}(a) shows the $A$ with a connected subgraph $B$ (the bold-line tree). Fig. \ref{NodeFinalJoin}(b) shows the contracted graph $A'$ with all vertices in $B$ merged to $s'$.
\end{Example}

OBSTA first builds a complete graph $G'$ consisting of the nodes in $D_{i,rest}^{l}$ and a new source $s'$, by contracting the nodes of $T_{i,stable}^{l}$ in $G$ to form a new source $s'$ and assigning each edge weight $e'_{u,v}$ in $G'$ as the shortest path length between nodes $u$ and $v$ in the contracted $G$.
Then, OBSTA constructs a minimum spanning tree $MST(G')$ rooted at $s'$ to span all destinations in $D_{i,rest}^{l}$ in $G'$, and then reverts $G'$ to $G$ to find the complete tree $T_{i,rest-MST}^{l}$. If $T_{i,rest-MST}^{l}$ has the sufficient deposit, OBSTA patches the candidate DT $T_i^l$ as $T_{i,rest-MST}^{l}$.
Otherwise, OBSTA constructs another tree $SPT(D_{i,rest}^{l})$ rooted at $s$ and extracted from $SPT(D)$ to span all destinations in $D_{i,rest}^{l}$.
Then, OBSTA grafts $SPT(D_{i,rest}^{l})$ on $T_{i,stable}^{l}$ (i.e., $T_{i,stable}^{l} \oplus SPT(D_{i,rest}^{l})$) to find the complete tree $T_{i,rest-SPT}^{l}$.
In this case, if $T_{i,rest-SPT}^{l}$ has sufficient deposit, OBSTA patches the candidate DT $T_i^l$ as $T_{i,rest-SPT}^{l}$. Otherwise, OBSTA discards the candidate DT $T_i^l$.

\begin{figure}[tb]
\centering
\includegraphics[scale=0.35] {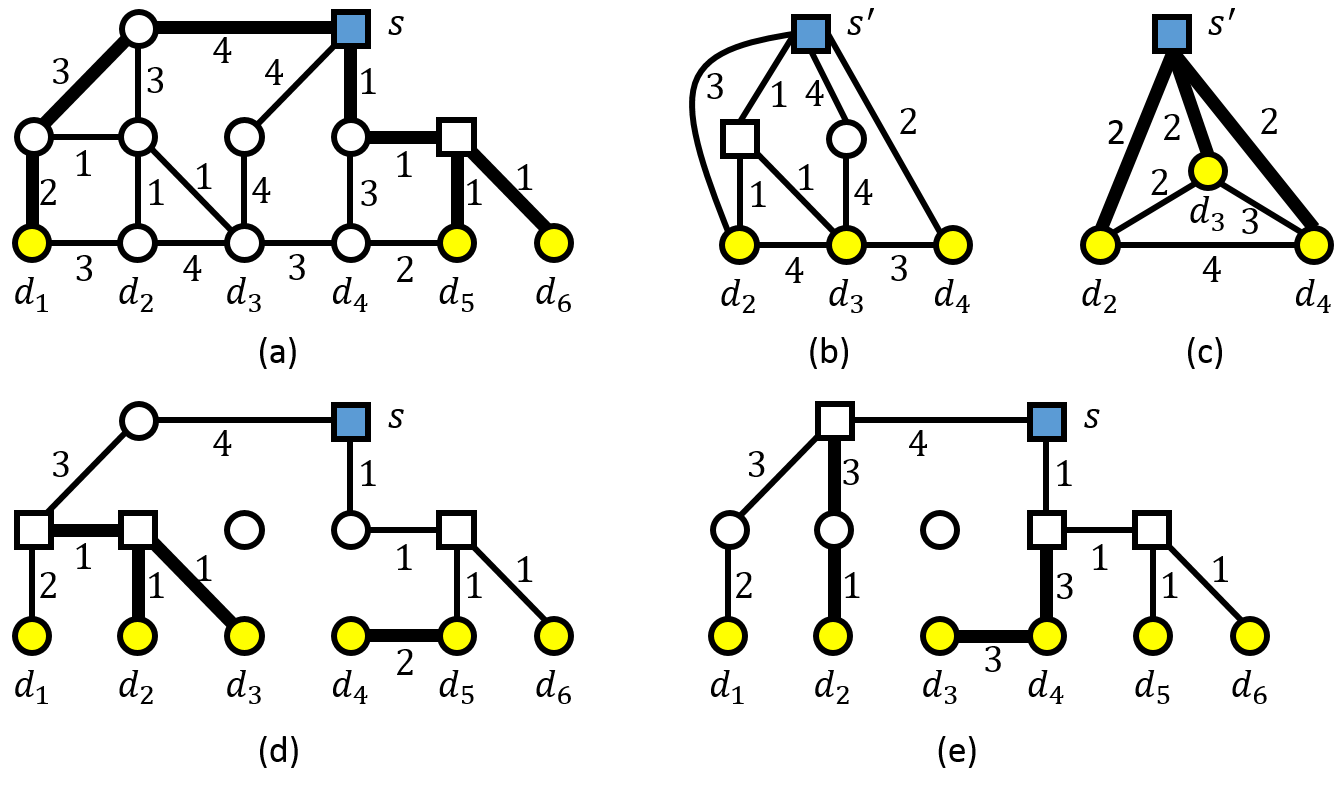}
\vspace{-1.5 mm}
\caption{Example of Candidate DT Patching. (a) Graph $G$ and candidate DT $T_{i,stable}^{l}$ (with bold lines). (b) Graph after contracting $T_{i,stable}^{l}$. (c) Graph $G'$ and Tree $MST(G')$ (with bold lines). (c) Tree $T_{i,rest-MST}^{l}$ obtained by recovering $MST(G')$. (d) Tree $T_{i,rest-SPT}^{l}$ obtained by $T_{i,stable}^{l} \oplus SPT(D_{i,rest}^{l})$.}
\label{NodeFinalJoin}
\end{figure}

\begin{Example}
Fig. \ref{NodeFinalJoin} presents an example of Candidate DT Patching for $d_2,d_3,d_4$ that have not been included in $T_{i}^l$ generated in Candidate DT Generation. Fig. \ref{NodeFinalJoin}(a) shows $G$ and the current candidate DT $T_{i,stable}^l$ including $d_1,d_5,d_6$. OBSTA first contracts $T_{i,stable}^l$ to $s'$ in Fig. \ref{NodeFinalJoin}(b) and constructs the complete graph $G'$ in Fig. \ref{NodeFinalJoin}(c). The constructed spanning tree $d_2,d_3,d_4$ (with bold lines) in $G'$ is in Fig. \ref{NodeFinalJoin}(c). Then, OBSTA reverts $G'$ to $G$ to obtain $T_{i,stable-MST}^l$ in Fig. \ref{NodeFinalJoin}(d). If $T_{i,stable-MST}^l$ does not have sufficient deposit, OBSTA constructs another tree $T_{i,rest-SPT}^l$ by $T_{i,stable}^{l} \oplus SPT(D_{i,rest}^{l})$ in Fig. \ref{NodeFinalJoin}(e).
\end{Example}
\subsubsection{Final Deployed Tree (DT) Selection}
This phase selects the final DT to 1) minimize the weighted sum of the tree cost, branch cost, and rerouting cost, 2) maximizes the deposit in $t_i$, and 3) effectively avoid the dramatic rerouting cost in the future. In other words, DT with the largest $g(T_i^{l})=\gamma \times dep_i^l
+ (1-\gamma) \times \beta\times prc_i^{l}$ is extracted as the solution tree,
where $prc_i^{l}$ is the potential rerouting cost of $T_{i}^{l}$, and $\gamma$ is a tunable parameter for the above factors. Intuitively, if the join and leave frequencies are not small, setting a small $\gamma$ can effectively reduce the potential rerouting cost.

\subsection{Analysis}
In the following, we first prove that at least one candidate DT has the sufficient deposit at each $t_i$ (i.e., $dep_i\geq \beta \times prc_i$). 
Subsequently, we show that $\sum_{i}{B_{i}}$ is no smaller than the solution $OBSTA(G)$ of OBSTA, and $\sum_{i}{B_{i}}$ is smaller than $|D_{max}|$ times of the optimal cost. 
Finally, we prove that OBSTA is a $|D_{max}|$-competitive algorithm.

\begin{lemma} \label{lemma: deposit constraint}
There is at least one candidate DT 
with the sufficient deposit
in each $t_i$.
\end{lemma}
\begin{proof}
The deposit in the beginning is $0$ and thus $dep_0\geq 0$. To prove by induction, we assume that the deposit at $t_{i-1}$ is non-negative and prove that the deposit at $t_i$ is non-negative (i.e., $dep_i\geq 0$), where $i\geq 1$.

Candidate DT Generation creates at least one candidate DT $T_{i}^{0}$ with only the source $s$ and deposit $dep_{i}^{0}=dep_{i-1}+B_{i}^{0}-w(T_{i}^{0})-\alpha\times b_{T_i}^{0}-\beta \times rc_{i}^{0}$. 
Note that $B_{i}^{0}=\alpha$, $w(T_{i}^{0})=0$, $b_{T_i}^{0}=1$, and $rc_{i}^{0}=0$. Thus, $dep_{i}^{0}=dep_{i-1}$. Moreover, because only $s$ is in $T_{i}^{0}$, there is no destination, and $prc_{i}^{0}=0$. Thus, $dep_{i}^{0}=dep_{i-1}\geq 0= \beta \times prc_{i}^{0}$.
Next, the first part of Candidate DT Patching can always generate a candidate DT ($T_{i,stable}^0$) 
with the sufficient deposit
because at least one candidate DT (e.g., $T_{i}^0$)
with the sufficient deposit
is produced by Candidate DT Generation.

Subsequently, we prove (by contradiction) that the candidate DT $T_{i}^{0}$ (i.e., $T_{i,rest-SPT}^{0}$) generated by the second part of Candidate DT Patching must
have the sufficient deposit.
Assume that candidate DT $T_{i}^{0}$ (i.e., $T_{i,stable}^0$) does not 
have the sufficient deposit.
That is,
there exists a set $D_{i,rest}^0$ such that $T_{i,stable}^{0} \oplus SPT(D_{i,rest}^{0})$
has the insufficient deposit.
Since $T_{i,stable}^0$ 
has the sufficient deposit,
we have two following inequalities representing the deposits before and after $T_{i,stable}^{0} \oplus SPT(D_{i,rest}^{0})$, respectively.
Note that we use hat symbol $\hat{}$ to mark the one after $T_{i,stable}^{0} \oplus SPT(D_{i,rest}^{0})$.
\begin{align}
dep_i^{0} &= dep_{i-1}+B_i^{0} - w(T_{i,stable}^{0}) - \alpha \times b_{T_i}^{0} - \beta  \times rc_i^{0}\notag\\
&\geq \beta \times prc_i^{0} \label{inequality: dep before}\\
\hat{dep}_i^{0} &= dep_{i-1}+\hat{B}_i^{0} - w(T_{i,rest-SPT}^{0}) - \alpha \times \hat{b}_{T_i}^{0} - \beta  \times \hat{rc}_i^{0}\notag\\
&<\beta \times \hat{prc}_i^{0} \label{inequality: dep after}
\end{align}
Since the budget increases by $\alpha \times|D_{i,rest}^0|+\sum_{x\in D_{i,rest}^0}dist_G(s,x)$, the tree cost grows by at most $\sum_{x\in D_{i,rest}^0}dist_G(s,x)$, and the branch cost increases by at most $|D_{i,rest}^0|$.
Thus, $\hat{prc}_i^{0} + \hat{rc}_i^{0} > prc_i^{0} + rc_i^{0}$.
We examine the generated cost for each node $d\in D_i$ as follows.
\begin{enumerate}
\item Node $d$ is in $T_{i,stable}^{0}$.
Both RT and $SPT(D_{i,rest}^0)$ are sub-trees rooted at $s$ in $SPT(D)$. Thus, $T_{i,stable}^{0} \oplus SPT(D_{i,rest}^{0})$ is also a sub-tree in $SPT(D)$. Since node $d$ is added according to RT, the potential rerouting costs for $d$ before and after $T_{i,stable}^{0} \oplus SPT(D_{i,rest}^{0})$ (i.e., $prc_i^{0}$ and $\hat{prc}_i^{0}$) are zero. The rerouting cost for $d$ remains the same after $T_{i,stable}^{0} \oplus SPT(D_{i,rest}^{0})$ because the routing path of $d$ does not change, leading to a contradiction. 
\item Node $d$ is in $D_{i,rest}^0$. After $T_{i,stable}^{0} \oplus SPT(D_{i,rest}^{0})$, $\hat{prc}_i^{0}$ becomes zero.
\begin{enumerate}
\item Node $d$ is in $D_{i-1}\cap D_{i}$. For $d$, $prc_i^{0}$ 
is no smaller than rerouting cost $\hat{rc}_i^{0}$
since $t_{i-1}$ finds $prc_i^{0}$ for $d$ from the symmetric difference between the path from $s$ to $d$ in $T_{i-1}$ and the shortest path from $s$ to $d$ in $G$, leading to a contradiction.
\item Node $d$ is not in $D_{i-1}$. For $d$, $rc_i^{0}=\hat{rc_i}^{0}=0$ in $t_i$. It leads to a contradiction since $prc_i^{0}$ is non-negative.
\end{enumerate}
\end{enumerate}
Thus, at least one candidate DT 
with the sufficient deposit.

Since all candidate DTs from Candidate DT Patching
have the sufficient deposit,
the one chosen by Final DT Selection also follows and generates a feasible tree.
\end{proof}
\begin{lemma} \label{lemma: budget sum no less OBSTA}
$\sum_{i}{B_{i}}$ is no smaller than $OBSTA(G)$
\end{lemma}
\begin{proof}
By Lemma \ref{lemma: deposit constraint},
the sufficient deposit guarantees 
$dep_{i}=\sum_{i}B_{i}-\sum_{i}(w(T_{i})+\alpha\times b_{T_i}+\beta \times rc_{i})\geq \beta \times prc_i$. Because 
$prc_{i}$ is non-negative in each $t_i$,
$$
\sum_{i}B_{i} \geq \sum_{i}(w(T_{i})+\alpha\times b_{T_i}+\beta \times rc_{i}).
$$
Therefore, $\sum_{i}B_{i} \geq OBSTA(G)$.
\end{proof}
\begin{theo} \label{theo: ratio}
OBSTA is a $|D_{max}|$-competitive algorithm.
\end{theo}
\begin{proof}
Let $\max_{d_i\in D_i} \{dist_{G}(s,d_{i})\}$ denote the maximum of all shortest paths from $s$ to $D_{i}$. 
\begin{align*}
\sum_{i}{B_{i}}&=\sum_{i}(\alpha \times |D_{i}|+\sum_{d\in D_{i}}dist_{G}(s,d))\\
&\leq \sum_{i}|D_{i}|(\alpha+\max_{d_i\in D_i} \{dist_{G}(s,d_{i})\})).
\end{align*}
Then, $\max_{d_i\in D_i} \{dist_{T_{i}^*}(s,d_{i})\}\geq \max_{d_i\in D_i} \{dist_{G}(s,d_{i})\}$ holds, where $T_{i}^*$ is the optimal tree at $t_i$, and 
\begin{align*}
\sum_{i}{B_{i}}&\leq \sum_{i}|D_{i}|(\alpha+\max_{d_i\in D_i} \{dist_{G}(s,d_{i})\})\\
&\leq \sum_{i}|D_{i}|(\alpha+\max_{d_i\in D_i} \{dist_{T_{i}^*}(s,d_{i})\})\\
&\leq |D_{max}|\sum_{i}(\alpha+\max_{d_i\in D_i} \{dist_{T_{i}^*}(s,d_{i})\})\\
&\leq |D_{max}|\sum_{i}(\alpha+w(T_{i}^*))\\
&\leq |D_{max}|\sum_{i}(w(T_{i}^*)+\alpha \times b_{T_i}^*+\beta \times rc_{i}^*),
\end{align*}
since $b_{T_i}^*\geq 1$. Also, $w(T_{i}^*)+\alpha \times b_{T_i}^*+\beta \times rc_{i}^*$ is the optimal cost at $t_i$, we have $\sum_{i}{B_{i}}\leq |D_{max}| \times OPT$.
\end{proof}

\noindent\textbf{Time complexity.}
Note that OBSTA takes time $O(|V|^2)$ to update the deposit, budget, rerouting cost, and potential rerouting cost, and examine whether the deposit is sufficient at each time slot. RT Construction takes time $O(|D|)$ to update the SI values of all nodes and then select the stable nodes (i.e., $SI(d)>H$) among them. Candidate DT Generation takes time $O(|V||D|)$, since it constructs at most $O(|D|)$ candidate trees and each of them needs time $O(|V|)$. For each candidate DT, Candidate DT Patching needs time {\color{black}$O(|V|^2|D|)$}, since it iteratively examines whether the deposit is sufficient after adding each of $O(|D|)$ stable nodes, where each addition and examination require time $O(|V|)$ and $O(|V|^2)$, respectively, and it then takes time $O(|V||D|+|D|^2)$, {\color{black}$O(|D|^2)$}, $O(|V||D|)$, $O(|V|^2)$, and $O(|V|^2)$ on graph construction, MST construction, graph reversion, SPT construction, and deposit update. Therefore, Candidate DT Patching takes time {\color{black}$O(|D|(|V|^2|D|))$}. Finally, Final DT Selection selects the best one among $O(|D|)$ candidate DTs, and thus the total complexity is {\color{black}$O(|V|^2|D|^2)$}. Note that the above result is worst-case analysis, while the next section indicates that OBSTA is very efficient for large networks with large groups (see Table \ref{t:running_time}).

\section{Performance Evaluation}\label{sec: evaluation}
We evaluate OBSTA by the simulation on real topologies with real multicast traces and implementation on an experimental SDN with HP SDN switches and YouTube traffic.

\subsection{Simulation Setup}
We compare OBSTA with SPT and ST\footnote{There are some single-tree static and dynamic multicast routing algorithms~\cite{UnicastTunnelForMulticast3}, \cite{reliableMulticastforSDN},  \cite{dynamic_multicast_mobile_ip}--\cite{dynamic_multicast_network_coding} with different purposes (such as QoS), but they are not included in this study because ST outperforms those approaches in terms of the bandwidth consumption, and the rerouting cost is not minimized in those approaches.} in two real networks, i.e., TataNld with 145 nodes/193 links and UsSignal from the Internet Topology Zoo with 63 nodes/78 links, as well as large synthetic networks with thousands of nodes to test the scalability. In the simulation, {\color{black}a link with the higher delay and higher loss rate (due to congestion) is assigned a larger weight \cite{LinkWeightSetting
}.} We set the packet loss rate of each link from 1\% to 10\% and link delay from 10 ms to 100 ms. Multicast join and leave dynamics are extracted from the trace file of SNDlib over 24 hours and divided into 195 time slots.\footnote{For large synthetic networks, the join and leave events are mapped to the synthetic nodes with the same IDs.}. $\alpha$ and $\beta$ are set to 0.1 and 0.6, respectively. We also change the duration of the time slots, the weight $\beta$ for rerouting, and the network size $|V|$. The performance metrics include the 1) total cost, 2) tree cost, 3) branch cost, and 4) rerouting cost defined in Section \ref{sec: related work}. We implement all algorithms in an HP DL580 server with four Intel Xeon E7-4870 2.4 GHz CPUs and 128 GB RAM. Each simulation result is averaged over 100 samples.

\subsection{Small Real Networks}
Fig. \ref{f:small_topo} compares OBSTA, SPT, and ST, and the tree costs (i.e., bandwidth consumption) of all approaches gradually increase since the joining events are much more than the leaving events. ST generates the smallest tree cost, which is much smaller than SPT. The tree cost of OBSTA is very close to ST, but the total rerouting costs of OBSTA are only 8.8\% of ST in TataNld and 9.3\% of ST in UsSignal, respectively, whereas SPT maintains shortest paths and thereby does not reroute and thereby is not shown in the figures. A larger rerouting cost incurs a higher volume of control messages and requires more CPU resources at both controller and switch sides~\cite{devoflow}. Moreover, the branch cost of OBSTA is much smaller than ST (at least 33\% reduction but not shown as a figure due to the space constraint) and thereby is more scalable for SDN. Therefore, OBSTA outperforms SPT and ST because it optimizes not only the tree cost and branch cost but also the rerouting cost at different time. 


\begin{figure}[t!]
\centering
\subfigure[]{
    \label{f:small_topo_a}
    \includegraphics[width = 4.2 cm]{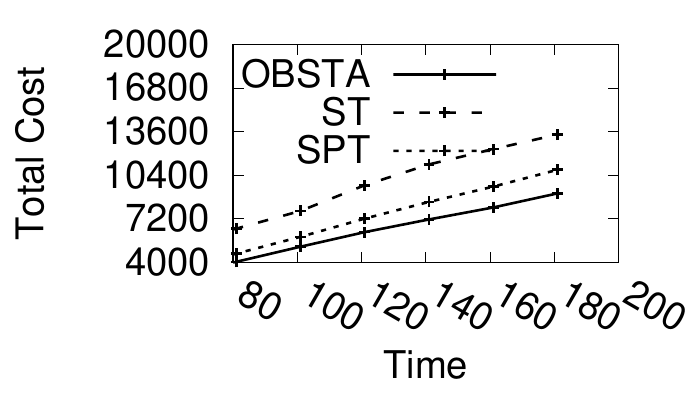}}
\subfigure[]{
    \label{f:small_topo_b}
    \includegraphics[width = 4.2 cm]{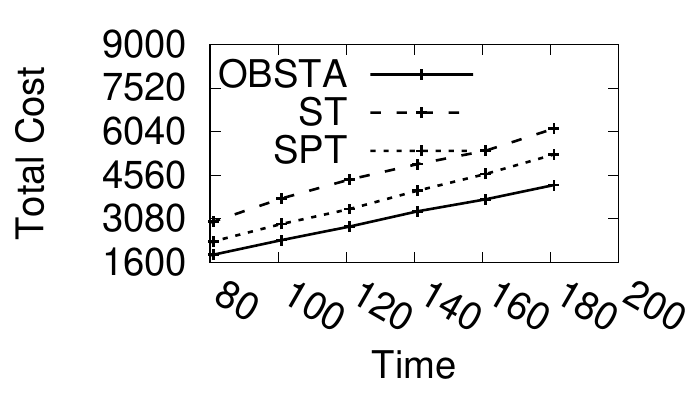}}
\\ \vspace{-4mm}
\subfigure[]{
    \label{f:small_topo_c}
    \includegraphics[width = 4.2 cm]{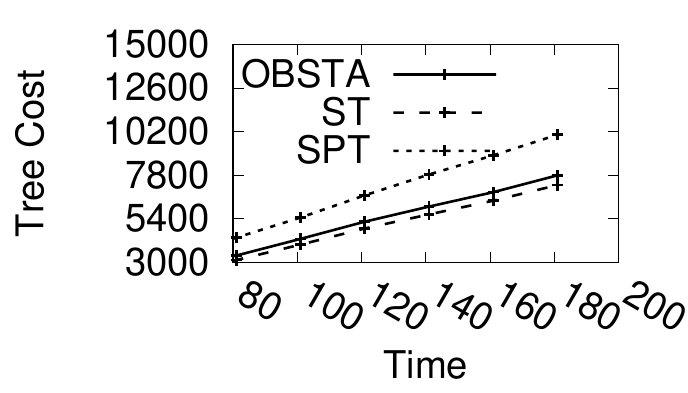}}
\subfigure[]{
    \label{f:small_topo_d}
    \includegraphics[width = 4.2 cm]{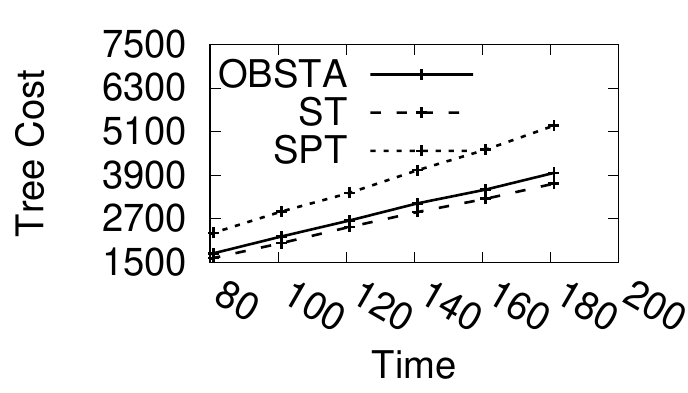}}
\\ \vspace{-4mm}
\subfigure[]{
    \label{f:small_topo_e}
    \includegraphics[width = 4.2 cm]{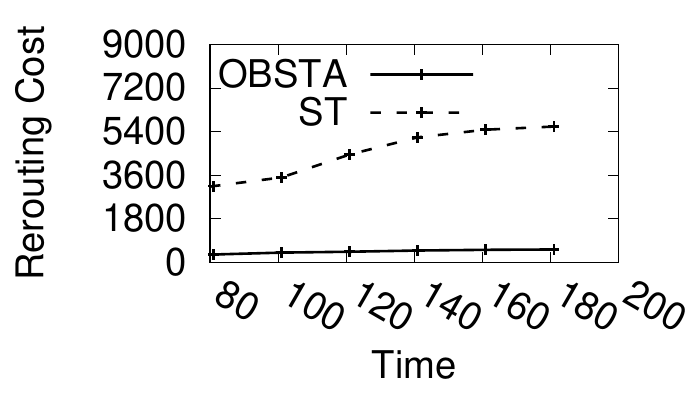}}
\subfigure[]{
    \label{f:small_topo_f}
    \includegraphics[width = 4.2 cm]{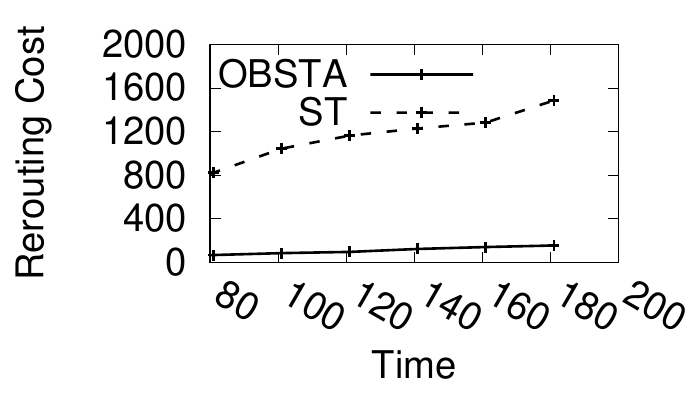}}
\caption{Simulation results for real networks. (a) Total cost for TataNld. (b) Total cost for UsSignal (c) Tree cost for TataNld. (d) Tree cost for UsSignal. (e) Rerouting cost for TataNld. (f) Rerouting cost for UsSignal.}
\label{f:small_topo}
\end{figure}

\subsection{Large Synthetic Networks}
Fig. \ref{f:large_topo} compares OBSTA, ST, and SPT in larger synthetic networks, where $|V|$ ranges from 1000 to 5000. The toal cost of OBSTA is 25\% less than SPT and ST, and the rerouting cost is only 9.8\% of ST since the temporal information is carefully examined in OBSTA. Compared with Fig. \ref{f:small_topo}, the advantage of OBSTA is more significant here, especially for rerouting costs. Minimizing the rerouting cost is more critical in a large SDN because the latency between a controller and switches increases, and high latency may deteriorate user QoE in video services. With a higher $\beta$ in Fig. \ref{f:large_topo_d}, the rerouting cost can be further reduced for various network sizes. {\color{black} We also change the duration of time slots from 10 to 40 events. Because OBSTA adjusts multicast tree less frequently with longer time slots, the tree costs are slightly increased by 5.5\%.} Table~\ref{t:running_time} summarizes the running time of OBSTA and ST (in brackets) with different $|D|$. In smaller networks ($|V| = 1000$ and $|V| = 2000$), the running time for OBSTA is less than 1 second. Although larger $|D|$ and $|T|$ indeed increase the computation time, OBSTA is still much more efficient than ST. Therefore, the above results manifest that OBSTA is more practical to be deployed in SDN controllers.


\begin{figure}[t!]
\centering
\subfigure[]{
    \label{f:large_topo_a}
    \includegraphics[width=4.2cm]{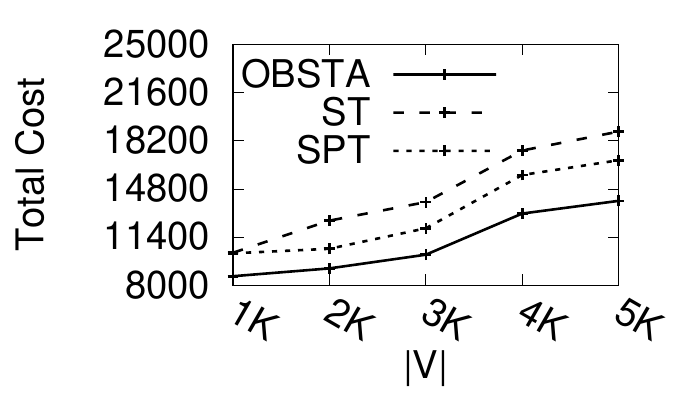}}
\subfigure[]{
    \label{f:large_topo_b}
    \includegraphics[width=4.2cm]{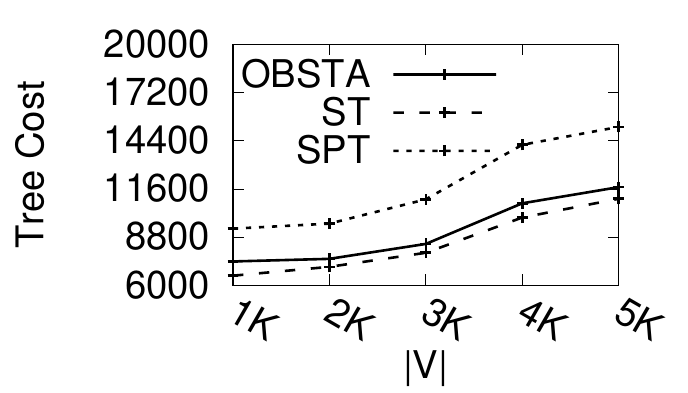}}
\\ \vspace{-4mm}
\subfigure[]{
    \label{f:large_topo_c}
    \includegraphics[width=4.2cm]{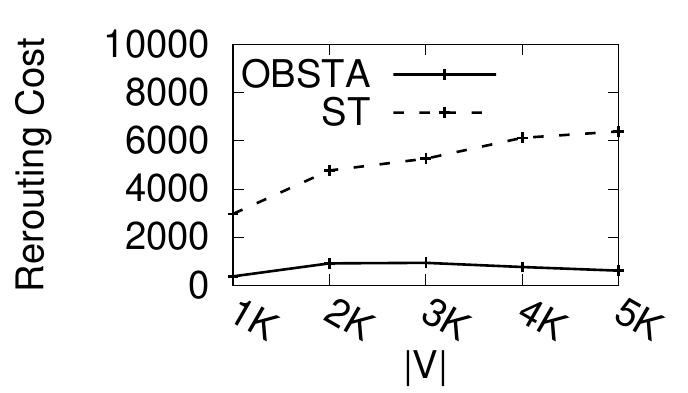}}
\subfigure[]{
    \label{f:large_topo_d}
    \includegraphics[width=4.2cm]{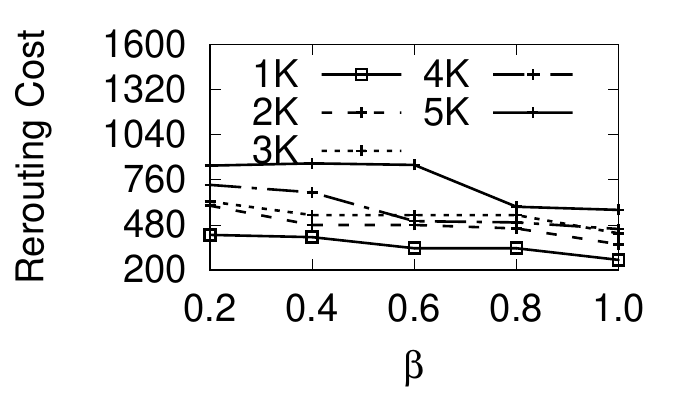}}
\caption{Simulation results for large synthetic networks. (a) Total cost with different $|V|$. (b) Tree cost with different $|V|$. (c) Rerouting cost with different $|V|$. (d) Total cost with different $\beta$ and $|V|$=5000.}
\label{f:large_topo}
\end{figure}


\begin{table}[t!]
\caption{The running time of OBSTA and ST (seconds per round)}
\begin{center}
\begin{tabular}{cccccc}
\hline
 $|D|$ & $|V|=1K$ & $|V|=2K$ & $|V|=3K$ & $|V|=4K$ & $|V|=5K$ \\ \hline
50 & 0.06(0.7) & 0.27(2.8) & 0.52(6) & 0.95(11.1) & 1.62(18.2) \\ \hline
100 & 0.13(1.4) & 0.55(5.9) & 1.03(12.7) & 1.91(23.3) & 3.27(38.2) \\ \hline
200 & 0.26(3) & 1.09(12.9) & 2.07(27.9) & 3.81(51.2) & 6.5(84) \\
\hline
\end{tabular}
\label{t:running_time}
\end{center}
\par
\end{table}

\subsection{Implementation}
We compare different algorithms in an experimental SDN with 20 HP Procurve 5406zl Openflow switches and 40 links, whereas OBSTA is implemented in an OpenDaylight SDN controller. A 191-second YouTube 4K video with 34.3 Mbps is first streaming to a proxy (since YouTube does not support multicast) and then multicasted to 15 destinations. Table~\ref{t:implementation} manifests that ST induces more rerouting and incurs much larger control overheads (20 times of OBSTA). A larger rerouting cost indeed increases the latency (more than 3 seconds) to install new rules into all involved switches for ST to complete the route setup. It may cause video stalling and deteriorate QoE. In contrast, the rerouting in OBSTA only induces 0.154 seconds. On the other hand, SPT produces larger multicast trees and consumes more network bandwidth. The result manifests that OBSTA is promising to support video services in SDN.

\section{Conclusion} \label{sec: conclusion}

Online multicast traffic engineering supporting IETF dynamic group membership is more practical but has not been studied for SDN. In this paper, we formulate Online Branch-aware Steiner Tree (OBST), to jointly minimize the bandwidth consumption, SDN multicast scalability, and rerouting overhead. We prove that OBST is NP-hard and does not have a $|D_{max}|^{1-\epsilon}$-competitive algorithm for any $\epsilon >0$, and the proposed $|D_{max}|$-competitive algorithm exploits the notion of the budget, the deposit and Reference Tree to achieve the tightest bound. 
The simulation and implementation on real SDNs with YouTube traffic manifest that
OBSTA can effectively reduce the total cost by 25\% compared with SPT and ST.

\begin{table}[t!]
\caption{Evaluation in the experimental SDN}
\begin{center}
\begin{tabular}{cccc}
\hline
 & OBSTA & ST & SPT \\
\hline
\#Ctrl. Msg. & 4 & 82 & 0 \\
\hline
Ctrl. Msg. Size & 592 B & 12136 B & 0 B \\
\hline
Ctrl. Latency & 0.154 s & 3.157 s & 0 s \\
\hline
BW Consumption & 3591 MB & 3375 MB & 4563 MB  \\
\hline
\end{tabular}
\label{t:implementation}
\end{center}
\end{table}

\bibliographystyle{IEEEtran}
\bibliography{mybibfile}
\appendices

\section{Pseudo-code}\label{sec: pseudo-code}

\begin{procedure} [h] \small
	\selectfont \caption{Sprouting $A\otimes (a_{1},a_{2},...,a_{n})$}
	\begin{algorithmic} [1]
		\REQUIRE  $A$, $(a_{1},a_{2},...,a_{n})$ and $s$
		\ENSURE  $(A_{1},A_{2},...,A_{n})$
		\STATE  $A_{0}\leftarrow \phi$
		\FORALL  {$1\leq i\leq n$}
			\STATE  $A_{i}\leftarrow A_{i-1}\cup P_{A}(s,a_{i})$
		\ENDFOR  
	\end{algorithmic}
\end{procedure}

\begin{procedure} [h] \small
	\selectfont \caption{Grafting $A\oplus B$}
	\begin{algorithmic} [1]
		\REQUIRE  $A$, $B$, $s$ and $d_{1}, d_{2}, ..., d_{n}$
		\ENSURE  $C$
		\STATE  $C\leftarrow A$
		\FORALL  {$1\leq i\leq n$}
			\STATE  $u \leftarrow d_{i}$
			\WHILE  {$u\notin A$}
				\STATE  $u\leftarrow u$'s parent in $B$.
			\ENDWHILE  
			\STATE  $C\leftarrow C\cup P_{B}(u,d_{i})$
		\ENDFOR  
	\end{algorithmic}
\end{procedure}

\begin{procedure} [h] \small
	\selectfont \caption{Reference Tree Construction}
	\begin{algorithmic} [1]
		\REQUIRE  $G=(V,E)$, $s$, $n$, $D$, $D_{1}$ to $D_{i}$ and $H$
		\ENSURE  $\tau_{i}$
		\FORALL  {$d\in D$}
			\STATE  $a_{d} \leftarrow$ the first arrival time slot of $d$
			\STATE  $h_{d} \leftarrow$ the duration in the group of $d$
			\IF {$h_{d}\leq \frac{n-a_{d}}{2}$}
				\STATE $SI(d)=\frac{2h_{d}}{n-a_{d}}$
			\ELSIF  {$h_{d}> \frac{n-a_{d}}{2}$}
				\STATE  $SI(d)=1$
			\ENDIF  
		\ENDFOR  
		\STATE  $\tau_i=|\{d| d\in D \wedge SI(d)> H\}|$
	\end{algorithmic}
\end{procedure}

\begin{procedure} [h] \small
	\selectfont \caption{Candidate Deployed Tree (DT) Generation}
	\begin{algorithmic} [1]
		\REQUIRE  $G=(V,E)$, $s$, $D_{i-1}$, $D_{i}$, $T_{i-1}$
		\ENSURE  $T_{i}^{l}$, $D_{i,join}^{l}$
		\STATE  $D_{i,join}\leftarrow D_{i}\backslash D_{i-1}$
		\STATE  $T_{i,leave}\leftarrow \cup_{d\in D_{i} \cap D_{i-1}} p_{T_{i-1}}(s,d)$
		\STATE  $T_{i}^{0}\leftarrow s$
		\STATE   $(q_{1},q_{2},q_{3},...,q_{k})\leftarrow $ the sorted sequence of destinations in a non-decreasing order of the between $s$ and the destination in $D_{i} \cap D_{i-1}$ in $T_{i-1}$.
		\STATE  $(T_{i}^{1},...,T_{i}^{k})\leftarrow T_{i,leave}\otimes (q_{1},q_{2},q_{3},...,q_{k})$
		\STATE  $D_{i,join}^{l}\leftarrow D_{i,join}\cup \cup_{j=l+1}^{k} q_{j}$
	\end{algorithmic}
\end{procedure}

\begin{procedure} [h] \small
	\selectfont \caption{Candidate Deployed Tree (DT) Patching}
	\begin{algorithmic} [1]
		\REQUIRE  $G=(V,E)$, $s$, $D_{i-1}$, $D_{i}$, $T_{i-1}$, $D_{i,stable}^{l}$, $T_{i}^{l}$, $D_{i,join}^{l}$
		\ENSURE  $T_{i}^{l}$
		\STATE  $k\leftarrow |D_{i,stable}^{l}|$
		\STATE  $(q_{1},q_{2},q_{3},...,q_{k})\leftarrow$ the sorted sequence of nodes $D_{i,stable}^{l}$ in a non-decreasing order of SI.
		\STATE  $(T_{i,stable}^{l_1},...,T_{i,stable}^{l_k})\leftarrow(RT_i(D_{i,stable}^{l})\otimes(q_{1},...,q_{k}))\oplus T_{i}^{l}$
		\STATE  $T_{i,stable}^{l}\leftarrow T_{i}^{l}$
		\FORALL  {$1\leq m\leq k$}
			\IF  {$T_{i,stable}^{l_m}$ has a sufficient deposit}
				\STATE  $T_{i,stable}^{l}\leftarrow T_{i,stable}^{l_m}$
			\ENDIF  
		\ENDFOR  
		\STATE  $D_{i,rest}^{l}\leftarrow D_{i}\setminus D_{i,stable}^{l}$
		\STATE  Build $T_{i,rest-MST}^{l}$
		\IF  {$T_{i,rest-MST}^{l}$ has a sufficient deposit}
			\STATE  $T_{i}^{l}\leftarrow T_{i,rest-MST}^{l}$
		\ELSE  
			\STATE  Build $T_{i,rest-SPT}^{l}$
			\IF  {$T_{i,rest-SPT}^{l}$ has a sufficient deposit}
				\STATE  $T_{i}^{l}\leftarrow T_{i,rest-SPT}^{l}$
			\ELSE  
				\STATE  Discards the candidate DT  $T_{i}^{l}$
			\ENDIF  
		\ENDIF  
	\end{algorithmic}
\end{procedure}

\begin{procedure} [h] \small
	\selectfont \caption{Final Deployed Tree (DT) Selection}
	\begin{algorithmic} [1]
		\REQUIRE  $T_{i}^{l}$, $dep_{i}^{l}$, $prc_{i}^{l}$ and $\gamma$
		\ENSURE  $T_{i}$
		\STATE  $g(T_{i}^{l})\leftarrow\gamma\times dep_{i}^{l}+(1-\gamma)\times prc_{i}^{l}$
		\STATE  $T_{i}\leftarrow \argmax_{T_{i}^{l}}\{g(T_{i}^{l})\}$
	\end{algorithmic}
\end{procedure}

\end{document}